\documentclass[12pt,hidelinks]{article}
\usepackage[english]{babel}
\usepackage{amssymb,amsmath,amsthm,amsfonts,mathtools}
\usepackage[top=2.8cm, bottom=2.8cm, left=2.8cm, right=2.8cm]{geometry}
\usepackage{graphicx}
\usepackage[utf8]{inputenc}
\usepackage{verbatim}
\usepackage{datetime}
\usepackage{setspace}
\usepackage{bookmark}
\usepackage{quoting} 
\usepackage{units}
\usepackage{xurl}
\usepackage{pgfplots}
\pgfplotsset{width=12cm,compat=1.9}
\usepackage{enumerate}   

\usepackage[ruled,vlined]{algorithm2e}
\usepackage[noend]{algpseudocode}
\SetKwComment{Comment}{/* }{ */}

\usepackage{natbib}

\usepackage{xcolor}

\definecolor{bblue}{HTML}{4F81BD}
\definecolor{rred}{HTML}{C0504D}
\definecolor{ggreen}{HTML}{9BBB59}
\definecolor{ppurple}{HTML}{9F4C7C}

\newtheorem{proposition}{Proposition}
\theoremstyle{definition}

\theoremstyle{proof}

\theoremstyle{definition}
\newtheorem{definition}{Definition}
\theoremstyle{remark}
\newtheorem{remark}{Remark}
\numberwithin{equation}{section}

\begin{document}
\title{Game Intelligence: Theory and Computation\footnote{I conducted this study without external funding and have no financial or non-financial conflicts of interest to disclose.}}
\author{Mehmet Mars Seven\thanks{Department of Political Economy, King's College London, London, WC2R 2LS, UK. mehmetmarsseven@gmail.com}}

\date{Revised: 26 October, 2025; First version: 9 December, 2022}
	\maketitle
\begin{abstract}
In this paper, I formalize intelligence measurement in games by introducing mechanisms that assign a real number---interpreted as an intelligence score---to each player in a game. This score quantifies the ex-post strategic ability of the players based on empirically observable information, such as the actions of the players, the game's outcome, strength of the players, and a reference oracle machine such as a chess-playing artificial intelligence system. Specifically, I introduce two main concepts: first, the \textit{Game Intelligence} (GI) mechanism, which quantifies a player's intelligence in a game by considering not only the game's outcome but also the ``mistakes'' made during the game according to the reference machine's intelligence. Second, I define \textit{gamingproofness}, a practical and computational concept of strategyproofness. To illustrate the GI mechanism, I apply it to an extensive dataset comprising over a billion chess moves, including over a million moves made by top 20 grandmasters in history. Notably, Magnus Carlsen emerges with the highest GI score among all world championship games included in the dataset. In machine-vs-machine games, the well-known chess engine Stockfish comes out on top. \textit{JEL Codes}: C72, C60, C81
\end{abstract}

\section{Introduction}
To win a chess game, top players do not always play the objectively best move. Instead, they often deviate from ``main lines'' to lure their opponents into uncharted territory, where human intelligence and creativity can shine. When both players follow established main lines, games often end in peaceful draws. However, deviating too far from the best move can be risky: a significant mistake greatly increases the chances of losing. Therefore, elite players, such as World No. 1 Magnus Carlsen, excel at balancing this trade-off, carefully weighing the risks of deviation against the potential rewards.

While intelligence has been extensively studied in psychology and computer science, it has received relatively little attention in economics, where there is a notable gap in the literature regarding a rigorous, game-theoretic framework for intelligence measurement. This paper addresses that gap by introducing mechanisms that assign a real number---interpreted as an intelligence score---to each player in a game. This score quantifies the ex-post strategic ability of the players based on empirically observable information, such as the actions of the players, the game's outcome, strength of the players, and a reference oracle machine such as a chess-playing artificial intelligence (AI) system.

I first introduce a novel framework for $n$-person general-sum games in which players may not possess full knowledge of certain aspects of the game, including actions, outcomes, and other players. As an example, playing an actual game of chess falls into this framework, even though the idealized version of chess is a game of perfect information. In this framework, I define a novel and practicable mechanism, called the \textit{Game Intelligence} (GI), that assigns an intelligence score to a player in a game by assessing not only the game's outcome but also the ``mistakes'' made relative to a reference machine's intelligence. Specifically, the machine assigns a probability distribution over the outcomes for each possible action a player might choose at any node. The GI mechanism is a function of the best moves predicted by the machine, the actual moves chosen by the player, moves of the other players, the game's outcome, and the Elo ratings of the players (if applicable). 

In the context of chess, the GI mechanism can be informally defined as:
\[
GI = f(\textit{Payoff}\,(\text{Moves}), \textit{Mistakes}\,(\text{Moves}, \text{Machine}, \text{Outcomes}), \text{Elo rating}),
\]
where $f$ is a real-valued function that is
\begin{enumerate}[(i)]
    \item increasing in the player's own payoff,
    \item decreasing in the magnitude of mistakes,
    \item increasing in the opponent's Elo rating.
\end{enumerate}
Put differently, the more points a player gains, the smaller the magnitude of mistakes made in the game, and the higher the opponent's Elo rating (if applicable), the greater the GI score. The function \textit{Payoff} outputs the player's payoff based on the sequence of moves played. \textit{Mistakes} is a real-valued function that measures the magnitude of the player's mistakes in the game. This measure depends on the chosen moves, a chess engine's evaluations, the results of engine self-play, and potential counterfactual outcomes.\footnote{In this paper, the term \textit{AI system} is used informally and broadly. The term \textit{machine} refers to a formalized version of an AI system, while \textit{engine} refers to a machine specifically in the context of chess.}

I illustrate the applicability of the GI mechanism by analyzing more than a billion chess moves played by humans and machines. Over a million of these moves are made by the world's top chess grandmasters. The dataset includes all world chess championship games from 1886 onwards. In this subset, Magnus Carlsen emerges with the highest GI score, despite not being the most accurate player according to top engine evaluations. When the analysis is extended to the larger dataset---which includes both world championship and elite tournament games---Bobby Fischer, Magnus Carlsen, and Garry Kasparov stand out as having the highest GI scores. The scores of these three grandmasters are statistically and practically indistinguishable from one another. 

With these findings, the GI mechanism arguably satisfies a rare combination of desirable properties: it has an intuitive theoretical definition with clear comparative statics implications, it is computationally feasible even for large datasets, and it produces empirically sound results. In addition to these main contributions of the paper, I briefly address three recurring theoretical questions concerning (i) the existence of plays with maximum intelligence scores, (ii) dynamic consistency, and (iii) strategyproofness.

First, I consider whether there always exists a play in $n$-person games with maximum intelligence scores, given a particular machine and mechanism. The answer is affirmative (Remark~\ref{rem:maximally_intelligent}); however, this result is largely of theoretical interest, as players typically do not have access to machines while they play games. Second, I define dynamic consistency, the property that a machine's evaluation remains unchanged after a move. While desirable in theory, this property is generally unattainable for finite machines in many real-world games. Nonetheless, dynamically consistent machines simplify computations for mechanisms like GI (Proposition~\ref{prop:GI_dynamic_consistency}).\footnote{I also explore a different concept of ``consistency'': a mechanism is called consistent if the intelligence score it produces does not conflict with either the game's outcome or the machine's evaluation. I show that both the GI mechanism and its variants satisfy this consistency notion (Proposition~\ref{prop:consistent}).} Finally, I define gamingproofness, a practical and computational version of strategyproofness, which holds if a player cannot increase her expected intelligence score by strategically choosing a suboptimal move. I show that the expected GI score is gamingproof under a natural assumption: that the player's estimation of the machine's deviation from her move is less accurate than her own evaluation of the deviation (Proposition~\ref{prop:gamingproofness}). Intuitively, this condition fails when the player either has access to a superior machine or has greater intelligence than the one used to evaluate her.

The rest of the paper is organized as follows. Section~\ref{sec:overview} presents an overview of the computational results. Section~\ref{sec:literature} reviews the relevant literature in psychology, computer science, and economics. Section~\ref{sec:setup} lays out the theoretical framework, while section~\ref{sec:theoretical_questions} addresses common theoretical questions. Section~\ref{sec:data} describes the data sources and computational methods and then reports the results. Finally, section~\ref{sec:conclusions} concludes.

\section{An overview of the computational results}
\label{sec:overview}

\begin{table}[t!]
\centering
\begin{tabular}{lcc}
\hline
\hline
Subset & Games & Moves \\ 
\hline
World Chess Championships & 2,006 & 86,737 \\
Top 20 players in history &  14,464 & 1,260,633 \\
Machine games &  1,743,565 & 226,470,027 \\
Regular human games &  20,157,126 & 1,343,129,524 \\
\hline
Total  & 21,915,155 & 1,570,860,184 \\
\hline
\end{tabular}
\caption{An overview of the dataset.}
\label{tab:overview_datasets}
\end{table}

Here, I present a summary of the computational results applied to more than one billion chess moves, including over a million moves made by the top 20 chess grandmasters in history. The full dataset, summarized by Table~\ref{tab:overview_datasets}, includes all official world chess championship games since 1886, numerous elite tournaments, regular human-vs-human games, and engine-vs-engine games.

To provide some background information, chess is a two-person game in which the players (White and Black) alternate turns, making moves until one player wins or the game ends in a draw. Chess moves can be analyzed using powerful AI systems, usually called engines, some of which are much stronger than the best human players. The analysis in this paper relies on evaluations by the leading chess AI, Stockfish.

GI scores are standardized so that the average is 100 and the standard deviation is 15. A higher GI score indicates a higher intelligence score. The GI score captures the trade-off between following the engine's top recommended move and deviating from it. A high GI score implies that the benefits of deviating from the engine's suggestions outweigh the potential costs. Conversely, a low GI score suggests that the player takes too many risks. 

Missed Points (MP) is a novel metric which is defined as the sum of a player's expected losses in a game. MP measures the difference between the perfect moves as determined by an engine and the actual moves made by the player. A lower MP suggests that a player's moves are closer to the engine's recommendations, meaning higher accuracy. Conversely, a higher MP indicates greater deviation from the engine's recommended moves. For example, an MP of 1.0 suggests that the deviations from perfect moves are equivalent to the loss of a game, considering that a win is worth 1 point. In contrast, an MP of 0 means that the player played a perfect game according to the engine. However, as mentioned in the Introduction, playing the top engine-recommended moves does not usually win games at the top level of chess. Since all elite players use engines to prepare, these top moves are commonly known, and they often lead to draws.

\subsection*{World Chess Championships}

\begin{table}[t!]
    \centering
\begin{tabular}{lcccc||lccccc}
\multicolumn{5}{c}{World Chess Championships} & \multicolumn{5}{c}{Top 20 grandmaster games} \\ \hline
Player & GI & MP & Games & Moves & Player & GI & MP & Games & Moves\\ \hline
Carlsen & 161 & 0.41 & 56 & 2958 & Fischer & 157 & 0.75 & 640 & 26,669 \\
Anand & 159 & 0.4 & 88 & 3323 & Carlsen & 157 & 0.66 & 1187 & 57,422 \\
Kasparov & 158 & 0.52 & 197 & 7076 & Kasparov & 157 & 0.71 & 1188 & 46,103 \\
$\vdots$ & $\vdots$ & $\vdots$ & $\vdots$ & $\vdots$ & $\vdots$ & $\vdots$ & $\vdots$ & $\vdots$ & $\vdots$ \\
Tal & 146 & 1.11 & 42 & 1906 & Cordoba & 143 & 1.07 & 72 & 2,938 \\
Euwe & 144 & 1.13 & 63 & 2709 & Vaganian & 143 & 1.15 & 54 & 2,383 \\
Steinitz & 143 & 1.3 & 115 & 4769 & Larsen & 140 & 1.23 & 80 & 3,084 \\
\hline
\end{tabular}
\caption{Summary of computational results for world chess championship games 1886--2023 (left) and top 20 grandmasters (right). The mean values of GI score and MP are presented.}
\label{tab:average_GMs_small}
\end{table}

Table~\ref{tab:average_GMs_small} (left) presents the highest GI scores of world chess champions from 1886 to 2023. Magnus Carlsen stands out with the highest GI score of 161, achieving this despite having a higher MP than Anand, who has the second highest average GI, closely followed by Kasparov. This suggests that Carlsen's playing style, while not always following the AI's optimal move, has a tendency to elicit more mistakes from his opponents compared to other players.\footnote{Unlike the GI scores for regular players, the GI scores for these grandmasters are not adjusted with respect to their opponents' Elo ratings, considering the widely recognized issue of rating inflation since the 1970s. This makes Carlsen's GI score even more impressive.}

William Steinitz, the first official world chess champion, has the lowest GI score, 143, closely followed by Max Euwe. They are also the two least accurate players among world champions. Their MPs indicate that their mistakes amount to slightly more than making a game-losing mistake on average.

\begin{figure}[t!]
    \centering
    \includegraphics[width=0.5\textwidth]{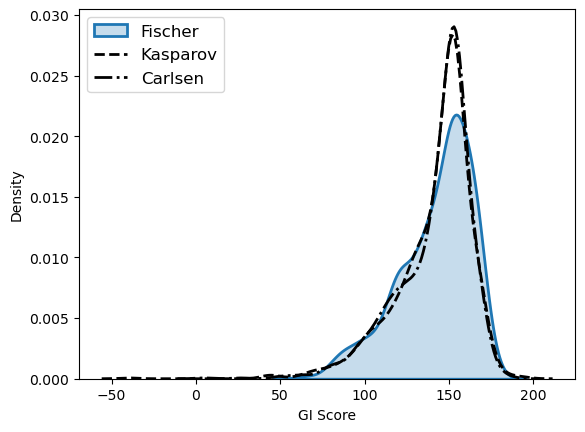}
    \caption{GI score distributions of Fischer, Carlsen, and Kasparov in elite competitions}
    \label{fig:fischer_carlsen_GIs}
\end{figure}

\subsection*{Top 20 grandmaster games}
Table~\ref{tab:average_GMs_small} (right) takes a broader perspective, summarizing the GI scores and MPs for ``super'' grandmasters (for more details, see Table~\ref{tab:average_GMs}). This dataset includes an analysis of over one million moves made in elite chess tournaments.

Fischer, Carlsen, and Kasparov stand out with the highest GI score of 157 (rounded). Their GI scores are statistically and practically indistinguishable. However, their MPs tell a different story, with Carlsen's MP (0.66) being lower than both Kasparov's (0.71) and Fischer's (0.75), suggesting that Carlsen's moves were closer to the AI's optimal choice. The differences in their MPs are statistically significant, as shown in Table~\ref{tab:mann-whitney_MP}. Interestingly, Fischer manages to secure the highest GI score, despite his relatively high MP. This suggests that even though Fischer's moves may not always coincide with the AI's top choice, they tended to provoke more mistakes from his opponents.

Figure~\ref{fig:fischer_carlsen_GIs} compares the GI score distributions for Fischer, Carlsen, and Kasparov in elite competitions. Carlsen's and Kasparov's distributions are remarkably similar. However, Fischer's distribution has a visibly different shape, which could be due to the fact that he stopped playing chess immediately after becoming a world champion.

\subsection*{Machine-vs-machine games}

The GI mechanism is also applicable for assigning intelligence scores to machines. Table~\ref{tab:average_engines} presents the top five GI scores of chess engines. These GI scores and MPs are computed from over 200 million moves in engine-vs-engine games, played under Computer Chess Rating Lists conditions. For details, see section~\ref{subsec:ccrl}.

The well-known engine Stockfish, when all of its versions combined, achieves the highest GI score of 172, followed by Komodo and Torch. Notice that the top engines each have a negative MP, which indicates that these engines are superior to the engines that evaluate them. In contrast, Clarity, the engine with the lowest GI score, has a positive MP, indicating that superior engines accurately identify Clarity's mistakes.

\begin{table}[t!]
\centering
\begin{tabular}{lcc||cccc}
\multicolumn{3}{c}{Best version of the engine} & \multicolumn{4}{c}{All versions combined} \\ \hline
Engine & GI & MP  & GI & MP  &  Games & Moves \\ \hline
\hline
Stockfish 150720 64-bit & 178 & -0.29 & 172 & -0.14  & 158,999 & 10,080,134 \\
Houdini 3 32-bit & 176 & -0.28  & 169 & -0.05  & 46,810 &  3,285,536 \\
Komodo 11.01 64-bit & 174 & -0.20 & 171 & -0.10  &  119,271 &  8,043,284 \\
Rybka 3 32-bit & 173 & -0.15  & 169 & -0.02  & 52,570 & 3,498,424 \\
Torch v1 64-bit & 171 & -0.09  & 171 & -0.09  & 3,307 & 216,237 \\
$\vdots$ & $\vdots$ & $\vdots$ & $\vdots$ & $\vdots$ & $\vdots$ & $\vdots$ \\
Clarity 2.0.0 64-bit & 140 & 0.28  & 157 & 0.24 & 1,297 & 72,938 \\
\hline
\end{tabular}
\caption{Summary of computational results for machine-vs-machine games. The mean values of GI score and MP are presented.}
\label{tab:average_engines}
\end{table}

\section{Related work}
\label{sec:literature}

The present study focuses on actual play in games, instead of complete plans of action (i.e., strategies). This approach is inspired by what can be called ``Bonanno's program'' (see, e.g., \citet{bonanno2013} and more recently \citet{bonanno2022}), which advocates a shift of emphasis towards describing actual play in games. While this paper draws inspiration from Bonanno's program, its framework is distinct, designed specifically to measure intelligence in games.

In games studied here, players do not have complete knowledge or information about aspects of the game, including other players, strategies, game tree, and outcomes. Although this study does not explicitly aim to model unawareness in the way it has been modeled, it does allow for the possibility that players may be---and usually are---unaware of certain aforementioned aspects of the game. For example, a player may be unaware of certain moves in a chess game. In that respect, I was inspired by the burgeoning body of research on unawareness that began with \citet{fagin1987}. Early studies on unawareness, such as those by \citet{fagin1987} and \citet{heifetz2006}, were not specifically based on game theory (though the latter study considers a multi-agent setting). However, game theory has been increasingly used in this field. For a review of the literature on unawareness, see \citet{schipper2014}.

In addition, there are game-theoretic frameworks, some of which were in part inspired by chess, which focus on equilibrium concepts and/or strategies~\citep{jehiel2007,halpern2014,ismail2025}. In contrast to these solution concepts and frameworks, the present study does not assume equilibrium, strategy, or rational play.

\subsubsection*{Backgammon, chess, and Go}
AI systems have been developed for games like backgammon and chess, including TD-Gammon \citep{tesauro1995} and GNU Backgammon for the former, and Deep Blue \citep{campbell2002} and Stockfish for the latter. These systems allow players to evaluate the quality of their moves and overall performance using metrics such as ``error rate'' in backgammon, and ``average centipawn loss'' and ``total pawn loss''~\citep{anbarci2024} in chess.

A centipawn is 1/100 of a pawn's value in chess and is used to evaluate a position by an engine. As is well-known, the main weakness of the pawn-based quality measures, such as the average centipawn loss metric, is that losing $c>0$ centipawns in a winning position is in general very different from losing $c$ centipawns in an equal position. However, traditional backgammon engines, some recent chess engines---such as Leela Chess Zero, which is inspired by AlphaZero \citep{silver2018}, and some Go engines such as AlphaGo Zero and its variants---adopt a different approach. These engines determine the best move in a position by maximizing the probability of winning a game, implicitly assuming that winning a game is worth 1 point and drawing is worth a half-win, thereby equating the expected score with the win probability. Unlike the earlier methods, MP is a metric that takes into account the relative weights of winning, drawing, and losing, and it applies to both zero-sum games and non-zero-sum games, where ``winning'' has no clear meaning.

A common limitation of the existing methods for analyzing games is that they are purely engine-based, where the evaluations of these engines assume that they play against themselves. Clearly, these metrics do not capture the ``human'' aspect of games. The main contribution of the GI score as compared to the existing literature is that it takes into account almost all empirically observable information in a game. This includes the game's outcomes, players' Elo ratings (where applicable), the relative weights of potential rewards, as well as the machine evaluations.

A considerable body of empirical research has used chess engines to assess the pawn-based quality of moves. While a comprehensive review of this literature is beyond the scope of this note, some selected works in the fields of computer science and economics include those by \citet{guid2006,regan2011,strittmatter2020,kunn2022,kunn2023,zegners2020,backus2023}.\footnote{The majority of these studies used the Stockfish chess engine, which only recently announced probabilistic evaluations in its latest version, released in December 2022.} Chess can serve as a valuable framework for examining cognitive performance and gender differences in competitiveness. For instance, a recent study by \citet{kunn2022} discovered a decline in performance when players competed remotely compared to in-person. In another example, \citet{sousa2022,sousa2023, backus2023} study gender differences in chess competitions.

\subsubsection*{Human and machine intelligence}
For over a century, psychometric tests have been widely used in psychology to evaluate an individual's cognitive abilities \citep{spearman1904}. These assessments typically measure skills such as verbal, mathematical, and spatial reasoning. For the literature on intelligence measurement in both human and artificial agents, refer to \citet{bartholomew2004,goertzel2007,bostrom2014,schaul2011,hernandez2017}.

Measuring machine intelligence has a long history in AI research. The Turing Test \citep{turing1950} is a classic method for evaluating intelligence through natural language conversations, while the Universal Intelligence Measure \citep{legg2007} offers a formal definition of intelligence and a quantitative measure based on Kolmogorov complexity and AIXI \citep{hutter2005}. Additionally, some alternatives have been proposed, such as the C-test and the anytime universal intelligence test, which aim to assess various aspects of intelligence \citep{hernandez1998,hernandez2010}.

While the aforementioned theoretical works define machine intelligence in terms of average performance (i.e., rewards) across various environments, my definition of game intelligence for a player encompasses not only rewards and opponents, but also the evaluation of the player's actions by a reference machine's intelligence. The novelty of the GI score is that it combines these elements in a unique way that has a clear interpretation, where players in a game can be ranked based on their GI score. Like many intelligence tests, GI can be applied to both humans and machines.

\subsubsection*{Strategyproofness and gamingproofness}

In strategic games and mechanism design, a key question is how resistant the mechanism is to manipulation. The concept of gamingproofness proposed in this paper is related to the concept of strategyproofness studied in social choice, mechanism design, and sports. Selected works that define strategyproofness include \citet{elkind2005,faliszewski2010,li2017,aziz2018,brams2018c} and \citet{pauly2013}. 

I have introduced gamingproofness as a simple computational version of strategyproofness for $n$-person games in this paper. Gamingproofness differs from the well-known notions of strategyproofness in that it uses the notion of play rather than a strategy profile as a primitive. In essence, a mechanism is gamingproof if players cannot increase their intelligence score by deliberately making a suboptimal move from their own perspective. If a player has a greater computing power or intelligence than the machine used to measure the player's intelligence, then the mechanism is typically not gamingproof. If a mechanism satisfies gamingproofness, then the intelligence scores become more reliable. Without gamingproofness, one might question whether the data truly represents a player's skill, or if the player is simply trying to artificially inflate their perceived intelligence, \`a la Goodhart's law.

\subsubsection*{Bounded rationality}
As mentioned above, the primary aim of this paper is to assess player intelligence in $n$-person games using game outcomes and evaluations from machines. While players may have varying degrees of rationality, my focus is not on classifying them as rational, boundedly rational, or irrational. Instead, I concentrate on evaluating their intelligence given their information and cognitive constraints. Bounded rationality has been a topic of study for a long time, but the concept has no single definition \citep{simon1955}. \citet{rubinstein1998} provides a comprehensive overview of the subject in his textbook. In this paper, the players could be considered ``boundedly rational.'' They have incomplete knowledge of important game aspects, such as the game tree and possible outcomes, and their limited foresight prevents them from performing full backward induction. Despite these limitations, some players in my setup can also be considered ``rational'' in some sense. They make the best decisions they can, given their information, environment, time, and computational constraints.

\section{The theoretical framework}
\label{sec:setup}
\subsection{Games}

Table~\ref{table:terminology} introduces the basic notation I use in this paper.
Let $\Gamma=(N, H, X, Z, I, \boldsymbol{\bar{A}}, R)$ be a game. $N=\{1,2,\dots,n\}$ denotes the finite set of players where $n\geq 1$, $X$ a finite game tree (i.e., the set of nodes), $x_0$ the root of the game tree, $h\in H$ an information set, $Z\subset X$ the set of terminal nodes, $I$ the player function, $R_i: Z\rightarrow \mathbb{R}$ the payoff or ``reward'' function of player $i\in N$.\footnote{The reward function could be a von Neumann-Morgenstern utility function \citep{neumann1944} if it is observable or measurable in the context of the game being played. In that case, the risk attitude of the players would be captured in these utilities.} I do not assume any specific form (e.g., zero-sum) for the reward function.\footnote{Unlike a standard extensive form game with imperfect information, $\Gamma$ does not contain strategy profiles and does not require knowledge of the aspects of the game, as I discuss below. For a standard textbook on game theory, see, e.g., \citet{osborne1994}.}

 \begin{table}[htb!]
 \centering
\begin{tabular}{|c|c|c|}
\hline
\textbf{Name} & \textbf{Notation} & \textbf{Element} \\
\hline
Players & $N = \{1, 2, \dots, n\}$ & $i$ \\
Information sets & $H$ & $h$ \\
Nodes/positions & $X$ & $x$ \\
Terminal nodes & $Z$ & $z$ \\
Player function & $I : H \to N$ & \\
Actions at node $x$ & $A_i(x)$ & $a_i(x)$ \\
Subplays & $\bar{A}$ & $\bar{a}$ \\
Plays & $\boldsymbol{\bar{A}} \subseteq \bar{A}$ & $\boldsymbol{\bar{a}}$ \\

Outcomes/rewards & $R = \{r_1, r_2, \dots, r_{t}\}$ & $r_j \in \mathbb{R}$ \\
Player $i$'s reward from $\boldsymbol{\bar{a}}$ & $R_i(\boldsymbol{\bar{a}})$ & \\
Probability distribution over $R$ & $\Delta R$ & $p$ \\
Machines & $\mathcal{M}$ & $M$ \\
$i$th component of machine $M$ & $M_i : X \to \Delta R$ & \\
Machine-optimal action & $a^*_i(x)$ & \\
Game & $\Gamma = (N, H, X, Z, I, \boldsymbol{\bar{A}}, R)$ & \\
$i$th component of mechanism $\mu$ & $\mu_i : \bar{A} \times \mathcal{M} \to \mathbb{R}$ & \\

\hline
\end{tabular}
\caption{Notation for the game framework}
\label{table:terminology}
\end{table}

$I:H \to N$ denotes the player function, where $I(h)$ gives the ``active'' player who moves at information set $h$. For each player $i$, $A_i(h)$ represents the set of possibly mixed actions (i.e., behavior strategies) over the finite set of pure actions at information set $h$. For convenience, the following notational conventions are used throughout the text. $I(x)$ and $A_i(x)$ refer to $I(h)$ and $A_i(h)$, respectively, whenever $x\in h$. In addition, $a_i(x)$ refers to both the action $a_i$ taken at node $x$ and to the successor node reached by taking that action. Since each action leads to a unique successor node, and for every node other than $x_0$ there is exactly one action leading to it, this convention does not cause ambiguity.

For a node $x_m \in X$, the \textit{path} (of play) between $x_0$ and  $x_m$ is denoted by
\[
[x_m]=\{x_0, x_1, x_2, \dots, x_m\},
\]
where for every $j=0,1,\dots,m-1$, $x_{j+1}$ is an immediate successor of $x_{j}$. Let $[x_m]_i\subseteq [x_m]$ denote the set of nodes at which player $i$ is active. 

For a node $x_m$, a \textit{subplay} between $x_0$ and $x_m$ is denoted by $\bar{a}$, which is a sequence of actions $a_i(x)$ chosen by player $i=I(x)$ at node $x$ on the path $[x_m]$. A subplay $\bar{a}$ is called a \textit{play} and denoted by $\boldsymbol{\bar{a}}$ when $x_m$ is a penultimate node, i.e., the final decision node prior to a terminal node. Let $\bar{A}$ denote the set of all subplays and $\boldsymbol{\bar{A}} \subseteq \bar{A}$ the set of all plays.

The finite set of outcomes or rewards in the game is denoted by $R=\{r_1, r_2, \dots, r_{t}\}$, where $r_j\in \mathbb{R}$ is a typical element, with $t$ being the total number of outcomes. For example, in chess, there are three possible outcomes: win, draw, and loss, with their corresponding values being 1, 1/2, and 0, respectively. However, the value of an outcome is not necessarily common among the players. With slight abuse of notation, the function $R_i(\boldsymbol{\bar{a}})$ denotes player $i$'s outcome resulting from play $\boldsymbol{\bar{a}}$, as each play uniquely determines a terminal node, and each terminal node is associated with a unique play. Let $\Delta R$ denote the set of all probability distributions over the set of outcomes $R$ and $p\in \Delta R$ be a probability distribution: that is, $\sum^{t}_{j=1}p(r_j)=1$, and for all $j$, $p(r_j)\geq 0$. A player may have a probability distribution over a proper subset $R'$ of the set of outcomes $R$, partly because the player might not know the complete set of outcomes. With a slight abuse of notation, this is captured by a probability distribution $p\in \Delta R$ such that none of the elements in $R\setminus R'$ are in the support of $p$.

\subsubsection*{Machines}
Define a (stochastic) \textbf{machine} $M_i$, $i\in N$, as a function $M_i: X\rightarrow \Delta R$ such that for every information set $h\in H$ and every $x$ and $y$ in $h$, $M_i(x) = M_i(y)$. Let $M$ denote a profile of machines.\footnote{I do not assume that machine $M_i$ is a player, but one can define a player that always plays according to $M_i$'s optimal choice.}

In simple terms, a machine yields a probability distribution, $M_i(x)$, over a subset of outcomes $R$ for node $x\in X$ in a game. For a node $x\in X$, $M_i(x)(r_j)$ denotes the machine $M_i$'s probability mass on the outcome $r_j\in R$ where $j\in\{1,2,\dots,t\}$. Let $\mathcal{M}$ be the set of all machines. 

Note that the machine function does not necessarily ensure that the probability distribution $M_i(x)$ is ``consistent'' with the potential outcomes given node $x$. Consequently, certain machines may assign a probability distribution to some outcomes that seem unreasonable or counterintuitive. An example of a machine is a chess engine which takes in a position and outputs a probability distribution associated with the position.

Note also that a machine is an abstract term that can refer to a human player, denoted as $M^h_i$, in a game. This can be interpreted as follows: For any node $x\in X$, a human player has a \textit{prior} (i.e., probability distribution) $M^h_i(x)\in \Delta R$. 

In the context of games where a specific game-theoretic solution concept is well-defined, the solution (i.e., a specific strategy profile) may be interpreted as a machine.\footnote{Note that classical solution concepts such as Nash equilibrium are not, in general, applicable to the games studied here. In a sense, a machine is conceptually a generalization of a solution concept.} For example, in a standard extensive form game, a Nash equilibrium would induce a probability distribution in every information set.  These induced distributions would then define a machine in that specific game.

\begin{definition}[Machine-optimal actions]
An action $a^*_i(x)$ is called \textbf{machine-optimal} with respect to machine $M_i$ if it satisfies the following condition:
\[
a^*_i(x)\in \arg\max_{a_i(x)\in A_i(x)} \sum^{t}_{j=1} r_j M_i(a_i(x))(r_j).
\]
\end{definition}
In simple words, a machine-optimal action at a node is the action with the highest expected value according to machine $M_i$.\footnote{As noted earlier, the notation $a_i(x)$ refers to both the action taken at node $x$ and the (possibly terminal) node that the action leads to.} A \textbf{human-optimal} action is simply an action that is machine-optimal with respect to $M^h_i$. In other words, a human-optimal action gives the highest subjective evaluation at a position with respect to the human's prior.

\subsubsection*{Players' knowledge} Here, I do not assume that players have common knowledge of each other's priors. Specifically, given an information set $h\in H$, player $i=I(h)$ knows that it is his or her turn to choose an action, but may not know all the elements of the information set $h$, nor the set of all available actions, $A_i(h)$. Players also do not necessarily know their own or other players' future available actions, or the outcome of every play. As noted above, a player might not know the complete set of possible outcomes in a position.

Players' knowledge about the game, including actions, outcomes, and other players, is captured by their subjective ``evaluation.'' Formally, player $i$'s \textbf{evaluation} of action $a_i(x)$ at node $x$ with respect to machine $M_i$ is defined as:
\begin{align*}
EV_i(a_i(x), M_i)=\sum^{t}_{j=1}r_j M_i(a_i(x))(r_j). 
\end{align*}

In this framework, the question of whether players have perfect recall is not directly relevant. Note that this does not negate the potential benefits a player with perfect recall may have within a game. However, the mechanisms in focus here operate from an ex-post perspective, similar to the post-game analysis done in chess. 

\subsection{Mechanisms}
\label{subsec:mechanisms}
Fix a game $\Gamma$. A (stochastic) \textbf{mechanism} is a profile $\mu$ of functions whose $i$th component is defined as $\mu_i:\bar{A} \times \mathcal{M}  \rightarrow \mathbb{R}$. For a given machine $M\in \mathcal{M}$, player $i$, and subplay $\bar{a}$, the function $\mu_i$ outputs a real number. The interpretation of $\mu_i$ is that it is a measure of human intelligence in reference to a machine's intelligence. Put differently, $\mu_i(\bar{a},M)$ assigns an \textbf{intelligence score} (i.e., a real number) to player $i$ for subplay $\bar{a}$ given reference machine $M$.\footnote{When Elo ratings are available, denoted by the vector $\mathcal{E} = (e_1, e_2, \dots, e_n)$ with $e_i \in \mathbb{R}_+$, one may consider rating-dependent mechanisms, denoted by $\mu_i^{\mathcal{E}}$. These mechanisms could incorporate player strength into the measurement of intelligence and will be used in the empirical analysis presented later in the paper.}

\subsubsection*{Practicable mechanisms}
The aim of this section is to define practicable mechanisms that are computationally feasible to implement on large-scale datasets, such as the one introduced in the Introduction. I begin by introducing the concept of ``missed points.''

\begin{definition}[Missed points]
Let $\bar{a}\in \bar{A}$ be a subplay and $[x_m]_i$ be its path comprising only the nodes at which $i$ is active. For node $x\in[x_m]_i$, let $a_i(x)$ be $i$'s action in subplay $\bar{a}$ and $a^*_i(x)$ be a machine-optimal action.  

Then, player $i$'s \textbf{missed points} (MP) is defined as
\[
MP_i(\bar{a}, M)=\sum_{x\in[x_m]_i} \left(EV_i(a^*_i(x), M_i) - EV_i(a_i(x), M_i)\right).
\]
\end{definition}

In words, the MP of player $i$ in subplay $\bar{a}$ is calculated by summing over each node $x$ the expected loss from choosing action $a_i(x)$. This expected loss is the difference between the evaluation of the best action $a^*_i(x)$ and that of the action actually chosen $a_i(x)$, according to the reference machine $M_i$. Thus, MP quantifies the total expected loss from suboptimal moves made by player $i$ during the subplay. The MP mechanism therefore is a measure of game accuracy relative to the reference machine.

We are now ready to define the following mechanism.

\begin{definition}[Raw Game Intelligence]
\label{def:game_intelligence}
Let $\boldsymbol{\bar{a}}\in \boldsymbol{\bar{A}}$ be a play. Then, the raw \textbf{game intelligence} (GI) score of $\boldsymbol{\bar{a}}$ for player $i$ is a mechanism $GI_i:\boldsymbol{\bar{A}} \times \mathcal{M}  \rightarrow \mathbb{R}$
defined as 
\[
GI_i(\boldsymbol{\bar{a}}, M) = R_i(\boldsymbol{\bar{a}}) - MP_i(\boldsymbol{\bar{a}}, M).
\]
\end{definition}

Recall that a subplay on the path $[x_m]$ is called a play if $x_m$ is a penultimate node. The raw GI of a play $\boldsymbol{\bar{a}}$ is calculated as the difference between a player's payoff, denoted $R_i(\boldsymbol{\bar{a}})$, and his or her MP. The intuition behind this mechanism is that it captures the trade-off between a player's actual gain and the expected value they missed during the game. 

For notational convenience, when a machine $M$ is fixed in a game $\Gamma$, $GI_i(\boldsymbol{\bar{a}})$ and $MP_i(\boldsymbol{\bar{a}})$ will denote $GI_i(\boldsymbol{\bar{a}}, M)$ and $MP_i(\boldsymbol{\bar{a}}, M)$, respectively. 

When more empirical information is available, one may consider a weighted variation of this mechanism. For example, if Elo ratings of players are known, the raw GI score can be adjusted accordingly. In general, the raw GI score is standardized to produce the GI score, as described in section~\ref{sec:population_average}. Note that the raw GI mechanism is, for ease of notation, purposefully defined only over plays. Its definition can, however, be extended to subplays as follows.

The following notion of ``expected'' game intelligence will be used in section~\ref{subsec:gamingproofness} on gamingproofness.

\begin{definition}[Expected GI]
\label{def:expectedGI}
Fix a game $\Gamma$ and a machine $M$. Let $\bar{a}$ be a subplay on the path $[x_m]$, and let player $i=I(x_m)$ be the active player at $x_m$. Then, the \textbf{expected game intelligence} (EGI) of $\bar{a}$ for player $j\in N$ is defined as 
\[
EGI_j(\bar{a}, M) = EV_j(a_i(x_m), M^h_j) - MP_j(\bar{a}, M).
\]
\end{definition}

The only difference between the EGI mechanism and the raw GI is that EGI replaces player $j$'s actual payoff, which may not have materialized because the game may not yet have finished, with the expected payoff $EV_j(a_i(x_m), M^h_j)$ at the node $a_i(x_m)$. This is calculated with respect to player $j$'s own prior $M^h_j(a_i(x_m))$ at $a_i(x_m)$.  The MP term is evaluated with respect to the machine $M$ as usual.

\section{Response to common theoretical questions}
\label{sec:theoretical_questions}
This section addresses several recurring theoretical questions that have emerged throughout the development and presentation of this paper.

\subsection{Machine-optimality vs. maximal intelligence}

Here, I discuss whether and under which conditions machine-optimal actions and ``maximally intelligent'' plays exist. 

\begin{remark}[Existence of machine-optimal actions]
    \label{rem:machine_optimality}
In every $n$-person game $\Gamma$, at every non-terminal node $x$ and for every machine $M$, there exists a machine-optimal action.
\end{remark}

\begin{proof}
Let $x$ be a node and $i=I(x)$. The following set,
\[
\arg\max_{a_i(x)\in A_i(x)} \sum^{t}_{j=1} r_j M_i(a_i(x))(r_j),
\]
is always nonempty because it maximizes a continuous expected value function over a compact set of actions.
\end{proof}

An important theoretical question is whether there exists a play with the ``highest'' intelligence score, hence the following definition.

\begin{definition}[Maximal intelligence]
\label{def:maximal_intelligence}
Fix a game $\Gamma$, a machine $M$, and a mechanism $\mu$. Let $\bar{a}$ be a subplay on the path $[x_m]$ and let player $i=I(x_m)$ be the active player at $x_m$. Player $i$'s action $a_i(x_m)$ in subplay $\bar{a}$ has \textit{maximal intelligence} at $x_m$ if for every subplay $\bar{a}'$ on the path $[x_m]$ that differs from $\bar{a}$ only at node $x_m$, the following inequality is satisfied:
\[
\mu_i(\bar{a},M)\geq \mu_i(\bar{a}',M).
\]
A subplay $\bar{a}$ has \textbf{maximal intelligence} if for every node $y\in [x_m]$, where $I(y)=j$, player $j$'s action $a_j(y)$ has maximal intelligence.
\end{definition}

Analogous to machine-optimality, an action with maximal intelligence is the one that maximizes the active player's intelligence score at that node. However, maximally intelligent and machine-optimal actions need not coincide even when the machine is stronger than the players themselves. As noted earlier, top chess players often deviate from the machine's recommended move to gain an advantage against their opponents.  

The following remark is immediate. 

\begin{remark}[Existence of maximally intelligent plays]
\label{rem:maximally_intelligent}
In every $n$-person game $\Gamma$ with perfect information, for every machine $M$ and every mechanism $\mu$, there exists a play that has maximal intelligence.
\end{remark}

This remark can be extended to games with imperfect information with the additional assumption of continuity of the mechanism in the actions chosen at each node.

While it is theoretically possible to find maximally intelligent plays in games, this approach is generally impractical in real-world human competitions. This is because players do not typically have knowledge of the machine and the mechanism used to calculate the intelligence score.

There is also a more demanding way to define maximal intelligence in games. Informally, a play ``maximizes intelligence'' if, at every node, players each choose an action that maximizes their intelligence score assuming that every player will do the same at all future nodes. This definition is in the spirit of subgame perfect Nash equilibrium, and in games with perfect information, its existence can be established using a backward induction argument. However, applying this concept in practice would not only require knowledge of the machine and the mechanism but also knowledge of the entire game tree. Yet, players and even machines often lack such complete knowledge.
 
\subsection{Gamingproofness}
\label{subsec:gamingproofness}

A common theoretical question in mechanism design is to consider whether players can manipulate the mechanisms. To address this potential concern and capture a simple and computational notion of strategyproofness suitable for this framework, I next define the concept of ``gamingproofness.''

Let $\hat{M}^h_i: X\rightarrow \Delta R$ denote a \textit{hypothetical machine} where $\hat{M}^h_i(x)$ is the probability distribution that human player $i$ believes machine $M_i$ assigns at node $x$. Note that $\hat{M}^h_i$ differs from $M^h_i(x)$, which denotes a human player's prior at $x$.

\begin{definition}[Gamingproofness]
\label{def:gamingproof}
Let $\bar{a}$ be a subplay on the path $[x_m]$ such that $a_i(x_m)$, where $i=I(x_m)$, is human-optimal at node $x_m$. A mechanism $\mu$ in a game $\Gamma$ is called \textbf{gamingproof} for player $i$ at node $x_m$ if player $i$'s action $a_i(x_m)$ has maximal intelligence with respect to the hypothetical machine $\hat{M}^h_i$. 

A mechanism $\mu$ is called \textit{gamingproof} if it is gamingproof for every player $i$ at every node $x\in X$ such that $I(x)=i$.
\end{definition}

A gamingproof mechanism ensures that at every node all players' human-optimal actions must have maximal intelligence with respect to their belief about the reference machine. In other words, no player can increase his or her intelligence score by intentionally making a suboptimal move, as viewed from their own perspective.

\begin{proposition}
\label{prop:gamingproofness}

Fix a game $\Gamma$ and a machine $M$. Let $\bar{a}$ be a subplay on the path $[x]$ such that $a_i(x)$, where $i=I(x)$, is human-optimal at node $x$.

The EGI mechanism is gamingproof for player $i=I(x)$ at node $x\in X$, if for every $a'_i(x)$ the following inequality holds:
\[
EV_i(a_i(x), M^h_i) - EV_i(a'_i(x), M^h_i)  \geq  EV_i(a'_i(x), \hat{M}^h_i) - EV_i(a_i(x), \hat{M}^h_i).
\]
\end{proposition}

\begin{proof}
By definition, the EGI mechanism is gamingproof for player $i$ at node $x$ if the human-optimal action $a_i(x)$ has maximal intelligence with respect to the hypothetical machine $\hat{M}^h_i$. Formally, for every subplay $\bar{a}'$ on the path $[x]$ that differs from $\bar{a}$ only at node $x$, the following must be satisfied:
\[
EGI_i(\bar{a}, \hat{M}^h_i)\geq EGI_i(\bar{a}', \hat{M}^h_i).
\]
By definition of EGI, this inequality holds if and only if
\[
EV_i(a_i(x), M^h_i) - MP_i(\bar{a}, \hat{M}^h_i) \geq
EV_i(a'_i(x), M^h_i) - MP_i(\bar{a}', \hat{M}^h_i).
\]

Observe that the MP of player $i$ before node $x$ is the same in the calculation of $EGI_i(\bar{a}, \hat{M}^h_i)$ and $EGI_i(\bar{a}', \hat{M}^h_i)$. Therefore, the inequality can be written as
\begin{align*}
& EV_i(a_i(x), M^h_i) - EV_i(a^*_i(x), \hat{M}^h_i) + EV_i(a_i(x), \hat{M}^h_i) \\
& \geq EV_i(a'_i(x), M^h_i) - EV_i(a^*_i(x), \hat{M}^h_i) + EV_i(a'_i(x), \hat{M}^h_i),
\end{align*}
where $a^*_i(x)$ is the machine-optimal action at $x$ with respect to the hypothetical machine $\hat{M}^h_i$. The inequality further simplifies to
\[
EV_i(a_i(x), M^h_i) -EV_i(a'_i(x), M^h_i)  \geq  EV_i(a'_i(x), \hat{M}^h_i) - EV_i(a_i(x), \hat{M}^h_i),
\]
as desired.
\end{proof}

The intuition behind Proposition~\ref{prop:gamingproofness} is based on the relationship between the players' evaluation of a particular position and their belief about the machine's evaluation of the same position. Consider a player's deviation $a'_i(x)$ from the human-optimal move $a_i(x)$ at position $x$. If, according to the player's belief $\hat{M}^h_i$, the machine's evaluation of this deviation is less than or equal to the player's own evaluation of the deviation, then the EGI mechanism is gamingproof at $x$.

In other words, a human player cannot strategically play a suboptimal move to increase their intelligence score, provided that their own evaluation of the position is superior to what they perceive as the machine's evaluation. This premise holds true only if the human's judgment is more accurate than their own approximation of the machine's judgment.

\subsection{Dynamic consistency of machines}

A common theoretical question I have received is what happens when a machine's evaluations are ``dynamically consistent.'' Here, I explore this question in the setting of two-player ``constant-sum alternating-move'' games. A game is constant-sum if, for every play $\boldsymbol{\bar{a}}\in \boldsymbol{\bar{A}}$, $R_1(\boldsymbol{\bar{a}})+R_2(\boldsymbol{\bar{a}})=c$, for some constant $c$. A two-player alternating-move game is one in which players take turns making moves in alternating order, starting from player 1.

\begin{definition}
    \label{def:dynamic_consistency}
    Let $\Gamma$ be a two-player constant-sum game.
    A machine $M$ on $\Gamma$ is called \textbf{dynamically consistent} if, for every node $x\in X$, the following conditions are satisfied.
    \begin{enumerate}
        \item For every player $i$ and every machine-optimal action $a^*_i(x)$, $M_i(x)=M_{i}(a^*_i(x))$.
        \item 
        \[
        \sum_{i\in N} EV_i(x,M_i) = c.
        \]
    \end{enumerate}
\end{definition}
Dynamic consistency requires that a machine's evaluation of a position remain unchanged after the machine's own recommended action is taken. Moreover, the machine's evaluation of a position should be consistent with the eventual outcome of the game. For example, in chess, if a machine evaluates a position as a win for both White and Black, it would not be dynamically consistent.

While dynamic consistency sounds like a desirable theoretical property, in my opinion, it is neither a necessary nor a sufficient condition for a ``good'' machine. For instance, a machine could mistakenly evaluate each position in a chess game as a win for White and loss for Black and still be dynamically consistent. On the other hand, Stockfish is an example of a machine that achieves superhuman performance despite not being dynamically consistent. Nevertheless, if a machine is dynamically consistent, then the relationship between the players' GI scores becomes simpler, as shown next.

\begin{proposition}[GI and dynamic consistency]
    \label{prop:GI_dynamic_consistency}
Let $\Gamma$ be a two-person alternating-move constant-sum game and fix a machine $M$. Let $\boldsymbol{\bar{a}}\in\boldsymbol{\bar{A}}$ be a play and $[x_m]$ be its path.

\begin{enumerate}
    \item If $[x_m]$ has an even number of nodes:
\begin{equation}
\label{eq:GI1}
GI_1(\boldsymbol{\bar{a}})= GI_2(\boldsymbol{\bar{a}})+2R_1(\boldsymbol{\bar{a}}) - EV_1(a^*_1(x_0), M_1) - EV_2(a_2(x_{m}), M_2) 
\end{equation}
    \item  If $[x_m]$ has an odd number of nodes:
\begin{equation}
\label{eq:GI2}
    GI_1(\boldsymbol{\bar{a}})= GI_2(\boldsymbol{\bar{a}})+R_1(\boldsymbol{\bar{a}})-R_2(\boldsymbol{\bar{a}})- EV_1(a^*_1(x_0), M_1) + EV_1(a_1(x_{m}), M_1)
\end{equation}
\end{enumerate}
\end{proposition}

\begin{proof}
Dynamic consistency of $M$ and constant-sum property of $\Gamma$ 
imply that for every node $x\in [x_m]_1$ on the path of $\boldsymbol{\bar{a}}$ at which player 1 is active, 
\[
EV_1(a_1(x), M_1) + EV_2(a_1(x), M_2) =  R_1(\boldsymbol{\bar{a}})+R_2(\boldsymbol{\bar{a}}). 
\]
By part (1) of dynamic consistency, $EV_2(a_1(x), M_2) = EV_2(a^*_2(a_1(x)), M_2)$.\footnote{Recall that $a_1(x)$ refers to both player 1's action at node $x$ and the successor node the action leads to. Thus, $a^*_2(a_1(x))$ denotes player 2's machine-optimal action at the node reached after player 1 takes action $a_1(x)$.} Thus, 
\begin{equation}
\label{eq:GI3}
EV_1(a_1(x), M_1) + EV_2(a^*_2(a_1(x)), M_2) =  R_1(\boldsymbol{\bar{a}})+R_2(\boldsymbol{\bar{a}}). 
\end{equation}
Analogously, dynamic consistency implies that for every node $y\in [x_m]_2$,
\begin{equation}
\label{eq:GI4}
EV_2(a_2(y), M_2) + EV_1(a^*_1(a_2(y)), M_1) =  R_1(\boldsymbol{\bar{a}})+R_2(\boldsymbol{\bar{a}}). 
\end{equation}

Recall that $GI_i(\boldsymbol{\bar{a}}) = R_i(\boldsymbol{\bar{a}}) - MP_i(\boldsymbol{\bar{a}})$. By definition of the GI scores:
 \[
GI_1(\boldsymbol{\bar{a}})=R_1(\boldsymbol{\bar{a}}) -\sum_{x\in [x_m]_1} \Bigl(EV_1(a^*_1(x), M_1) - EV_1(a_1(x), M_1)\Bigr).
\]
and
 \[
GI_2(\boldsymbol{\bar{a}})=R_2(\boldsymbol{\bar{a}}) -\sum_{x\in [x_m]_2} \Bigl(EV_2(a^*_2(x), M_2) - EV_2(a_2(x), M_2)\Bigr).
\]
Subtracting the second equation from the first leads to:
 \[
GI_1(\boldsymbol{\bar{a}})=GI_2(\boldsymbol{\bar{a}})+R_1(\boldsymbol{\bar{a}}) -R_2(\boldsymbol{\bar{a}})
\]
\[
+\sum_{x\in [x_m]_2} \Bigl(EV_2(a^*_2(x), M_2) - EV_2(a_2(x), M_2)\Bigr)
\]
\[
-\sum_{x\in [x_m]_1} \Bigl(EV_1(a^*_1(x), M_1) - EV_1(a_1(x), M_1)\Bigr).
\]

By Equations~\ref{eq:GI3} and Equation~\ref{eq:GI4}, the terms $EV_1(a_1(x), M_1) + EV_2(a^*_2(a_1(x)), M_2)$ and $-EV_2(a_2(a_1(x)), M_2) - EV_1(a^*_1(a_2(a_1(x))), M_1)$ cancel each other out. There are two cases to consider.

If $[x_m]$ has an even number of nodes, player 2 is active at node $x_m$. Then, the remaining terms in the two sums, after the cancellation, simplify to: 
\[
- EV_1(a^*_1(x_0), M_1) - EV_2(a_2(x_m), M_2) + EV_1(a_1(x_{m-1}), M_1) + EV_2(a^*_2(x_m), M_2).
\]
Since $EV_1(a_1(x_{m-1}), M_1) + EV_2(a^*_2(x_m), M_2) = R_1(\boldsymbol{\bar{a}})+R_2(\boldsymbol{\bar{a}})$,  
\[
GI_1(\boldsymbol{\bar{a}}) = GI_2(\boldsymbol{\bar{a}}) + 2R_1(\boldsymbol{\bar{a}}) - EV_1(a^*_1(x_0), M_1) - EV_2(a_2(x_m), M_2).
\]
This is Equation~\ref{eq:GI1}. 

If $[x_m]$ has an odd number of nodes, player 1 is active at node $x_m$. Then, the remaining terms in the two sums above, after the cancellation, are: 
\[
- EV_1(a^*_1(x_0), M_1) + EV_1(a_1(x_{m}), M_1).
\]
Thus, Equation~\ref{eq:GI2} is obtained in this case.   
\end{proof}

These formulas show that, under dynamic consistency, the GI score difference can be computed directly from the results of the game, the evaluation of the first machine-optimal move, and the evaluation at the terminal node. This simplification eliminates the need to compute every intermediate evaluation along the path of play.

That said, dynamic consistency does not appear to be a desirable property in the real world. In fact, any reasonable machine with finite computing power must be dynamically inconsistent in some game, since it is always possible to design a game whose solution requires more computing power than the machine possesses. If the machine were dynamically consistent in such a game, then its evaluation at the root would be essentially ``arbitrary.''

\subsection{Consistency of mechanisms}
As mentioned above, it may be an unrealistic expectation for machines to be dynamically consistent in real-world settings. However, it is more feasible to require mechanisms to maintain a less demanding level of ``consistency'' that aligns with the outcomes of plays and machine evaluations.

In what follows, I provide a formal definition of consistency in this framework. 

\begin{definition}[Consistency]
\label{def:consistency}
Fix a game $\Gamma$ and a machine $M$. A mechanism $\mu$ is called \textbf{consistent} if for every player $i$ and every pair of plays $\boldsymbol{\bar{a}}$ and $\boldsymbol{\bar{b}}$ in $\boldsymbol{\bar{A}}$ the following three conditions are satisfied.
\begin{enumerate}
\item $R_i(\boldsymbol{\bar{a}})\geq R_i(\boldsymbol{\bar{b}})$ and $MP_i(\boldsymbol{\bar{a}}) \leq  MP_i(\boldsymbol{\bar{b}})$ implies that  $\mu_i(\boldsymbol{\bar{a}}, M_i) \geq \mu_i(\boldsymbol{\bar{b}}, M_i)$. 
\item Assume  $MP_i(\boldsymbol{\bar{a}})=MP_i(\boldsymbol{\bar{b}})$. $R_i(\boldsymbol{\bar{a}})\geq R_i(\boldsymbol{\bar{b}})$ if and only if $\mu_i(\boldsymbol{\bar{a}}, M_i)\geq \mu_i(\boldsymbol{\bar{b}}, M_i)$.
\item Assume $R_i(\boldsymbol{\bar{a}})= R_i(\boldsymbol{\bar{b}})$. $MP_i(\boldsymbol{\bar{a}})\leq MP_i(\boldsymbol{\bar{b}})$  if and only if $\mu_i(\boldsymbol{\bar{a}}, M_i)\geq \mu_i(\boldsymbol{\bar{b}}, M_i)$.
\end{enumerate}
\end{definition}

To put it simply, a mechanism is consistent if the value it assigns to a play contradicts neither the outcome of the play nor the MP evaluation of the machine that is used to calculate the value of the mechanism. 

The three conditions in the definition of consistency specify how the mechanism should behave in different situations. For instance, if a play results in worse evaluation for both the human and the machine compared to another play, then the intelligence score of the latter must be greater than the intelligence score of the former according to the mechanism. Conditions (2) and (3) apply a comparative statics logic: if one of the evaluations (either the human evaluation or the machine evaluation) is held fixed, then the mechanism should be consistent with the other evaluation.

There is only one case where the concept of consistency does not apply. This is when there are two plays, $\boldsymbol{\bar{a}}$ and $\boldsymbol{\bar{b}}$, where $\boldsymbol{\bar{a}}$ leads to a strictly worse evaluation for the machine than $\boldsymbol{\bar{b}}$, and $\boldsymbol{\bar{b}}$ leads to a strictly worse evaluation for the human than $\boldsymbol{\bar{a}}$. In this case, consistency does not place any restrictions on the intelligence score of these plays. The reason is that sometimes a move that is optimal for a human may not be optimal for a machine, and vice versa. I next show that the GI mechanism is consistent. 

\begin{proposition}
\label{prop:consistent}
Fix a game $\Gamma$ and a machine $M$. Any mechanism of the form $\alpha R_i(\boldsymbol{\bar{a}}) - \beta MP_i(\boldsymbol{\bar{a}}) + \delta$, where $\alpha>0$, $\beta>0$, and $\delta\in \mathbb{R}$, is consistent. 
\end{proposition}
    
\begin{proof}
Let $\boldsymbol{\bar{a}}$ and $\boldsymbol{\bar{b}}$ in $\boldsymbol{\bar{A}}$ be a pair of plays. First, assume that $R_i(\boldsymbol{\bar{a}})\geq R_i(\boldsymbol{\bar{b}})$ and $MP_i(\boldsymbol{\bar{a}})\leq  MP_i(\boldsymbol{\bar{b}})$. It implies that 
        \[
         R_i(\boldsymbol{\bar{a}}) - MP_i(\boldsymbol{\bar{a}}) \geq  R_i(\boldsymbol{\bar{b}}) - MP_i(\boldsymbol{\bar{b}}).
        \]
        Thus, $GI_i(\boldsymbol{\bar{a}}) \geq GI_i(\boldsymbol{\bar{b}})$. 
    
        Second, assume that $MP_i(\boldsymbol{\bar{a}})=MP_i(\boldsymbol{\bar{b}})$. Then,
     \[
         GI_i(\boldsymbol{\bar{a}}) - GI_i(\boldsymbol{\bar{b}}) = R_i(\boldsymbol{\bar{a}}) - MP_i(\boldsymbol{\bar{a}}) - R_i(\boldsymbol{\bar{b}}) + MP_i(\boldsymbol{\bar{b}}).
        \]
         \[
          = R_i(\boldsymbol{\bar{a}})  - R_i(\boldsymbol{\bar{b}}).
        \]
        Thus, $GI_i(\boldsymbol{\bar{a}}) \geq GI_i(\boldsymbol{\bar{b}})$ if and only if $R_i(\boldsymbol{\bar{a}})  \geq  R_i(\boldsymbol{\bar{b}})$.
    
        Third, assume that $R_i(\boldsymbol{\bar{a}})= R_i(\boldsymbol{\bar{b}})$. Then,
     \[
        GI_i(\boldsymbol{\bar{a}}) - GI_i(\boldsymbol{\bar{b}}) = MP_i(\boldsymbol{\bar{b}}) - MP_i(\boldsymbol{\bar{a}}).
        \]
    Therefore, $GI_i(\boldsymbol{\bar{a}}) \geq GI_i(\boldsymbol{\bar{b}})$ if and only if $MP_i(\boldsymbol{\bar{b}}) \geq MP_i(\boldsymbol{\bar{a}})$.
        
    Clearly, the above calculations remain true if $R_i$ and $MP_i$ are scaled by strictly positive real numbers. Thus, any mechanism given by $\alpha R_i(\boldsymbol{\bar{a}}) - \beta MP_i(\boldsymbol{\bar{a}}) + \delta$, where $\alpha>0$, $\beta>0$, and $\delta\in \mathbb{R}$, must be consistent. 
    \end{proof}

It is easy to see that purely machine-based metrics such as MP would not be consistent in part because they do not depend on the outcome of a game. Analogously, the intelligence measures discussed in section~\ref{sec:literature} would not be consistent.

\section{The dataset, computational methods, and results}
\label{sec:data}

As shown by Table~\ref{tab:overview_datasets}, the full dataset consists of about 1.5 billion chess moves played by top grandmasters, regular players, and machines in engine-vs-engine competitions. Overall, I estimate that the total execution time for completing all computational tasks was roughly 700 hours on a computer equipped with an Intel\textsuperscript{\textregistered} Core\textsuperscript{TM} i7-6700 CPU, 4 cores/8 threads. This estimate includes all tasks such as data collection and cleaning, annotating games with raw (i.e., centipawn) Stockfish evaluations, and computing the relevant metrics for every move and game.

\subsection{World championships and elite tournaments}

This subset consists of approximately 1.2 million individual moves played in approximately 15,000 games by world-class chess grandmasters. The games were downloaded in Portable Game Notation (PGN) format from pgnmentor.com. The dataset includes games from world championships and many elite tournaments played by top grandmasters. 

I used these PGNs to obtain Stockfish raw (centipawn) evaluations from the Lichess Masters Database. Although this database is not publicly available to download, Lichess.org allows users to search opening positions to identify games between top players. The database includes many games that have already been annotated with Stockfish evaluations. I wrote a script to isolate the first 20 moves of each game, then used the Lichess Masters Database API to match these sequences with existing records. I downloaded and separated those PGN files that were (at least partially) annotated. I completed the missing evaluations by running them through Stockfish 16 with a depth setting of 20.

I have evaluated a subset of roughly 18\% of these games with Stockfish 16, the state-of-the-art chess analysis engine as of July 2023, on a personal computer. This process involved annotating every position in these games using the Stockfish centipawn evaluations. To calculate the MP of the players in these games, I used the empirical win-draw-loss percentages as a function of centipawn loss. These percentages were derived from a dataset where Stockfish played against itself.

I compared my Stockfish evaluations with Lichess's Stockfish evaluations of around 87,000 positions. The results closely mirrored those reported in section~\ref{sec:overview}. To conserve computational resources, I continued using Lichess's Stockfish evaluations in the dataset, with the exception of those games that had missing evaluations or the games of Bobby Fischer, most of which were not part of the Lichess database.

Cleaning the data involved separating non-classical games, such as rapid, blitz, online, blindfold, simul, and exhibition matches, from the classical ones. Additionally, games played by Fischer, Carlsen, and Kasparov in which the Elo rating was below 2500 were omitted, as this rating is generally accepted as the threshold for grandmaster-level play. I expanded the dataset by downloading games from almost all major tournaments and included the games which had at least one top 20 player (or Karpov) according to all-time Elo rankings.

Interestingly, there were very few Bobby Fischer games in the Masters Database. Thus, I analyzed as many of his classical games as possible using Stockfish 16 with a depth of 20. The games of Fischer and the games with missing evaluations are the only entries in my dataset that do not use Lichess's Stockfish evaluations.

Once the dataset grew to more than a million moves, I decided I had reached a logical stopping point. This was to avoid overusing the computational resources of my university as well as the resources provided by Lichess, which is funded by user donations.

\subsubsection{The top 3 vs other grandmasters}

The results in elite competitions clearly separate the top 3 players Fischer, Carlsen, and Kasparov from others. To further analyze whether their GI score distributions are statistically greater than those of the other grandmasters, I used a one-sided Mann-Whitney U test.\footnote{Since a chess game is played by two players, the GI scores of the two players are not independent. However, the top 3 players, Carlsen, Fischer, and Kasparov did not play a classical game against each other as far as I am aware.} Table~\ref{tab:mann-whitney_GI} shows the p-values resulting from the tests under the null hypothesis that there is no difference in the distributions between the row player and the column player. 

For example, the p-value of 0.39 in the Carlsen-Kasparov cell indicates there is insufficient evidence to conclude that Carlsen's GI score is greater than Kasparov's. Similarly, neither Carlsen's nor Kasparov's GI scores are statistically significantly different from Fischer's GI scores. In contrast, the table reveals that the differences in GI scores of each of the top 3 players are statistically significant, which means that there is sufficient evidence that their GI scores are greater than those of the other grandmasters.

The p-values for MP distributions are shown in Table~\ref{tab:mann-whitney_MP} in the Appendix. While there is weak evidence that Kasparov's MPs are greater than Carlsen's, the evidence is stronger that Fischer's MPs are greater than Carlsen's.

\begin{table}[h!]
\centering
\begin{tabular}{|l|c|c|c|c|c|c|}
\hline
& Kasparov & Fischer & Kramnik & Anand & Karpov & Topalov \\ \hline
Carlsen  & 0.39 & 0.70 & 0.00 & 0.00 & 0.00 & 0.00 \\ \hline
Kasparov  &  & 0.73 & 0.00 & 0.00 & 0.00 & 0.00 \\ \hline
Fischer  & &  & 0.00 & 0.00 & 0.00 & 0.00 \\ \hline
Kramnik  & & &  & 0.28 & 0.01 & 0.00 \\ \hline
Anand  & & & &  & 0.03 & 0.00 \\ \hline
Karpov  & & & & &  & 0.00 \\ \hline
\end{tabular}
\caption{P-values obtained from one-sided Mann-Whitney U tests. Null hypothesis: there is no difference in the GI score distributions between the row player and the column player.}
\label{tab:mann-whitney_GI}
\end{table}

\subsection{Regular players}
\label{sec:population_average}

To standardize the raw GI scores, I first computed the sample average and standard deviation in a large dataset of regular players, downloaded from Lichess Open Database, which consists of all games except bullet games played by their users (slightly more than one million players).\footnote{While official world championships exist for classical, rapid, and blitz chess, bullet chess (including hyper-bullet) does not have such recognition. Therefore, bullet games were excluded from this dataset.} This dataset consists of two parts, the games played in August 2023 (part 1) and in September 2023 (part 2).  Table~\ref{tab:population_average} presents the average raw GI scores, MPs, and their standard deviations. To check the robustness of these numbers, I have computed the same statistics for each part separately. As the table shows, both the averages and the standard deviations are remarkably close in each part. 

Before standardizing the raw GI scores, I weighted them using the Elo ratings of the opponents since the dataset includes players from all levels. The weight decreases as the Elo rating of the opponent decreases, and vice versa. To provide some background, the Elo rating system is used to rank players based on their past performance. Given two players with Elo ratings $E_i$ and $E_j$, the expected score for player $i$, represented by $ES$, is computed using the following formula:
\[
ES(E_i,E_j) = \frac{1}{1+10^{(E_j-E_i)/400}}.
\]
Let $GI'^{\text{raw}}$ be the player's raw GI score before weighting. Then, the raw GI score $GI^{\text{raw}}$ is calculated as follows.
\begin{align*}
GI^{\text{raw}} = GI'^{\text{raw}} - (1-2\times ES(E_j,2800))~|GI'^{\text{raw}}|,
\end{align*}
where $|\cdot|$ is the absolute value operator. The intuition behind this formula is as follows: If the  player's opponent has an Elo of $E_j=2800$, then the initial raw GI score remains the same, as this rating corresponds to the top 99th percentile in Lichess. As the rating $E_j$ of the opponent decreases, so will $ES(E_j,2800)$, hence a lower raw GI score will be assigned for winning against weaker players.

Next, to adjust the mean to 100 and the standard deviation to 15, I use the modified Z-score formula as follows.
\begin{equation*}
    z_i = 15 \times \frac{(x_i-\mu)}{\sigma} + 100.
\end{equation*}
In this formula, $x_i$ represents individual data points. The multiplication by 15 and the addition of 100 are adjustments to scale and shift the data to the desired mean and standard deviation.

The overall transformation can be expressed in the form of a linear function $z_i = b + a \times x_i$. The parameters $a$ and $b$ in this context are determined as follows: $a = \frac{15}{\sigma}$ and $b = -\frac{15\mu}{\sigma} + 100$. These values of $a$ and $b$ are then used to transform each data point $x_i$ to the corresponding $z_i$, ensuring the normalized dataset has a mean of 100 and a standard deviation of 15: 
$GI = -\frac{15\mu}{\sigma} + 100 + \frac{15}{\sigma} GI^{\text{raw}}$. 

The combined dataset is processed to calculate player-specific statistics. I considered only the players with at least 50 games and an average MP of at least 0. This criterion rules out the ``superhuman'' players who managed to play at least as well as Stockfish in 50 games, which is practically impossible for a human. Applying this filter reduced the dataset to 797,410,344 moves. Among these players, the average raw GI score is -3.1027 and the average MP is 0.8085, where both values are rounded. As a result, I obtain the following linear equation that transforms the raw GI score to the GI score: 
\begin{align*}
GI =  157.57  + 18.55 \times GI^{\text{raw}}.
\end{align*}
Thus, the mean GI score is 100 with a standard deviation of 15, which is in line with how the Intelligence Quotient (IQ) is presented.\footnote{Hence, one may refer to the GI score in this context as ``chess intelligence'' or ``chess IQ.'' Similarly, the GI score in a specific game or sport can be referred to with the game's name.}
\begin{table}[t!]
\centering
\begin{tabular}{lcc|cc|cc}
& \multicolumn{2}{c|}{Combined} & \multicolumn{2}{c|}{Part 1} & \multicolumn{2}{c}{Part 2} \\ \hline 
& Games & Moves & Games & Moves & Games & Moves \\ \hline \hline 
& 20,157,126 & 1,343,129,524 & 6,946,736 & 463,057,202 & 13,210,390 & 880,072,322 \\ \hline  \\ \hline 
& Average & Std Dev & Average & Std Dev & Average & Std Dev \\ \hline
\hline
Raw GI & -3.084 & 2.810 & -3.079 & 2.806 & -3.087 & 2.812 \\
MP & 1.973 & 1.367 & 1.970 & 1.365 &  1.975 & 1.368 \\
\hline
\end{tabular}
\caption{Classical, rapid, and blitz chess games data summary statistics.}
\label{tab:population_average}
\end{table}

\subsection{Engine-vs-engine games}
\label{subsec:ccrl}

As summarized in Table~\ref{tab:average_engines}, this subset includes over 200 million moves played by engines under ``CCRL 40/15'' testing conditions. These conditions, formulated by the Computer Chess Rating Lists (CCRL) organization, ensure a level playing field for computer chess analysis. Games are played with a time limit of 40 moves in 15 minutes on an Intel i7-4770K processor.  To further promote fairness, the use of specialized engine opening books is prohibited, and a common opening book limited to a depth of 12 moves is used. 

Additionally, all the games are played with the ponder feature deactivated, which means that each engine ``thinks'' only on its turn. This implies that the standard way to compute the MP (as well as traditional metrics such as average centipawn loss) would produce meaningless results because the evaluations of two different machines will obviously be very different. To define the GI score and MP in this setting, I use the fact that an engine always chooses the top move according to its own evaluation and therefore records 0 MP on its own turn. Consequently, the expected loss of an engine's move (say, Stockfish) can be measured by the evaluations of the opponent engine's (say, Komodo) preceding move and the move that follows Stockfish's move. Essentially, each engine measures the MP of its opponent.

Definition~\ref{def:game_intelligence}, the raw GI score, is updated for this framework as follows. Fix a two-player alternating-move game $\Gamma$, with machines $M$ and $M'$ playing as player 1 and player 2, respectively.  Let $\boldsymbol{\bar{a}}\in \boldsymbol{\bar{A}}$ be a play on the path $[x_m]$, where each action is machine-optimal. 

Player 1's missed point from the action $a_1(x)$ at node $x\in X$, where $1=I(x)$, is defined as 
\[
EV_{1}\left(a_2(a_1(x)), M'_{1}\right) - EV_{1}\left(x, M'_{1}\right).
\]
Here, the term $EV_{1}\left(a_2(a_1(x)), M'_{1}\right)$ is the expected value for player 1, as evaluated by player 2's machine $M'$, after player 2 responds to player 1's move $a_1(x)$. The second term $EV_{1}\left(x, M'_{1}\right)$ is the expected value for player 1, also as evaluated by player 2's machine, before making move $a_1(x)$. In this way, player 1's MPs throughout the game are measured using a single, consistent reference machine. Player 2's MPs are computed analogously. GI scores are then derived from these MP values.

The main difference between the standard GI score and this adapted version is the reference point for measuring a miss. In the standard definition, a move's miss is measured against the counterfactual best move that the same engine would have made in that position. In the adapted definition, the miss is measured as the difference in the opponent engine's evaluation before and after the move. In effect, each engine's intelligence is measured using the intelligence of the other engine.

Table~\ref{tab:all_engines} and Table~\ref{tab:all_engines_best_version} show the GI rankings of the engines that played at least 1000 games under CCRL conditions. The top five engines are Stockfish, Houdini, Komodo, Rybka, and Torch, respectively. These engines each have a negative MP, which indicates that they perform better than the engines evaluating them. Note that in human-vs-human games, MPs are all positive, as Stockfish is significantly stronger than humans.

\section{Concluding remarks}
\label{sec:conclusions}

Intelligence is a complex construct that cannot be perfectly captured by the outcome of a single game, even in highly skill-based games like chess. If a player receives a large reward from the outcome of a game, this does not necessarily indicate high intelligence, as chance may have significantly affected the outcome. For instance, in backgammon, a player may make a poor move but still have a positive probability of winning the game. Furthermore, machine-optimal moves may not always be optimal from a human perspective, and conversely, human-optimal moves may not be machine-optimal. For example, top chess players occasionally play ``suboptimal'' moves to increase their probability of winning a game. Thus, to assess intelligence in games, one must take into account both human- and machine-based (counterfactual) decisions.

I proposed the Game Intelligence (GI) as a practical mechanism that has broad applicability and can be used to assign intelligence scores to players and teams across a wide range of $n$-person games, including sports, board games, and video games, as well as individual decision problems. Nevertheless, caution is warranted when implementing the GI mechanism in sports or games where AI systems have not yet achieved a level of intelligence sufficient for measuring human intelligence.

The GI mechanism can also be used to assign intelligence scores to AI systems themselves. A potential future research direction is to use GI scores to train AI agents, as they provide more refined feedback than the simple binary outcome of a game.

\bibliographystyle{chicago}

\section*{Appendix}

\section{Additional tables}

\begin{table}[t!]
\centering
\begin{tabular}{lcccccccccc}
\hline
Player & GI  & $\text{GI}_W$ & $\text{GI}_B$ & MP & $\text{MP}_W$ & $\text{MP}_B$ & $\#$Games & $\#$Moves \\ \hline
\hline
Fischer & 157 & 158 & 156 & 0.75 & 0.73 & 0.76 & 640 & 26,669 \\
Carlsen & 157 & 157 & 157 & 0.66 & 0.70 & 0.63 & 1,187 & 57,422 \\
Kasparov & 157 & 158 & 156 & 0.71 & 0.70 & 0.71 & 1,188 & 46,103 \\
Karjakin & 155 & 158 & 153 & 0.66 & 0.62 & 0.70 & 493 & 23,388 \\
Aronian & 155 & 156 & 155 & 0.67 & 0.70 & 0.63 & 641 & 30,264 \\
Kramnik & 155 & 156 & 155 & 0.74 & 0.79 & 0.70 & 1,250 & 56,758 \\
Anand & 155 & 155 & 154 & 0.75 & 0.78 & 0.72 & 1,529 & 64,173 \\
Giri & 154 & 155 & 153 & 0.70 & 0.73 & 0.68 & 440 & 19,963 \\
Leko & 154 & 155 & 153 & 0.72 & 0.71 & 0.73 & 443 & 20,170 \\
Mamedyarov & 154 & 157 & 151 & 0.72 & 0.63 & 0.81 & 437 & 18,230 \\
Karpov & 154 & 155 & 153 & 0.87 & 0.89 & 0.83 & 1,448 & 72,154 \\
Svidler & 154 & 156 & 151 & 0.71 & 0.67 & 0.75 & 453 & 17,484 \\
Ivanchuk & 153 & 154 & 152 & 0.79 & 0.78 & 0.81 & 718 & 29,450 \\
Nakamura & 153 & 156 & 149 & 0.79 & 0.68 & 0.91 & 467 & 21,807 \\
Gelfand & 153 & 153 & 152 & 0.79 & 0.81 & 0.78 & 559 & 23,236 \\
Caruana & 152 & 154 & 151 & 0.85 & 0.83 & 0.87 & 559 & 26,851 \\
Grischuk & 152 & 155 & 150 & 0.81 & 0.79 & 0.82 & 401 & 18,496 \\
Adams & 152 & 154 & 149 & 0.82 & 0.77 & 0.86 & 452 & 20,105 \\
Topalov & 150 & 153 & 148 & 0.90 & 0.80 & 0.98 & 695 & 30,015 \\
Shirov & 150 & 152 & 149 & 0.90 & 0.90 & 0.91 & 540 & 23,051 \\
Short & 148 & 151 & 144 & 0.99 & 0.90 & 1.12 & 400 & 18,479 \\
$\vdots$  & $\vdots$ & $\vdots$ & $\vdots$ & $\vdots$ & $\vdots$ & $\vdots$& $\vdots$ & $\vdots$ \\
Cordoba & 143 & 141 & 145 &1.07 & 1.19 & 0.96 & 72 & 2,938 \\
Vaganian & 143  & 146 & 139 &1.15 & 1.04 & 1.25 & 54 & 2,383 \\
Larsen & 140  & 140 & 140 & 1.23 & 1.18 & 1.27 & 80 & 3,084 \\
\hline
\end{tabular}
\caption{Summary of computational results for ``super'' grandmasters. Mean values of GI and MP are presented. The subscripts $W$ and $B$ refer to White and Black, respectively.}
\label{tab:average_GMs}
\end{table}

\newpage

\begin{table}[t!]
\centering
\begin{tabular}{lcccccccc}
\hline
Engine & GI & $\text{GI}_W$ & $\text{GI}_B$ & MP & $\text{MP}_W$ & $\text{MP}_B$ & $\#$Moves & $\#$Games \\ 
\hline
Stockfish & 172 & 174 & 170 & -0.14 & -0.19 & -0.09 & 10080134 & 158999 \\
Komodo & 171 & 173 & 169 & -0.10 & -0.14 & -0.05 & 8043284 & 119271 \\
Torch & 171 & 173 & 168 & -0.09 & -0.14 & -0.04 & 216237 & 3307 \\
Houdini & 169 & 171 & 167 & -0.05 & -0.08 & -0.02 & 3285536 & 46810 \\
Rybka & 169 & 169 & 168 & -0.02 & 0.02 & -0.05 & 3498424 & 52570 \\
Comet & 168 & 170 & 167 & -0.10 & -0.14 & -0.06 & 66173 & 1043 \\
Vitruvius & 168 & 170 & 167 & -0.06 & -0.10 & -0.02 & 295741 & 4354 \\
NagaSkaki & 168 & 170 & 167 & -0.08 & -0.11 & -0.06 & 98912 & 1573 \\
Berserk & 168 & 170 & 166 & -0.04 & -0.09 & 0.01 & 1664702 & 25215 \\
BBChess & 168 & 169 & 168 & -0.07 & -0.08 & -0.06 & 87059 & 1369 \\
$\vdots$  & $\vdots$ & $\vdots$ & $\vdots$ & $\vdots$ & $\vdots$ & $\vdots$& $\vdots$ & $\vdots$ \\
Clarity 	& 157	& 158	& 157	& 0.24	& 0.22 & 	0.26 &	72938 &	1297 \\
\hline
\end{tabular}
\caption{Summary of computational results for the best versions of engines in the engine-vs-engine dataset. Means of GIs and MPs are presented.}
\label{tab:all_engines}
\end{table}

\begin{table}[t!]
\centering
\begin{tabular}{lcccccccc}
\hline
Engine & GI & $\text{GI}_W$ & $\text{GI}_B$ & MP & $\text{MP}_W$ & $\text{MP}_B$ & $\#$Moves & $\#$Games \\ 
\hline
Stockfish 150720 64-bit & 172 & 177 & 168 & -0.29 & -0.37 & -0.21 & 37844 & 558 \\
Houdini 3 32-bit & 170 & 171 & 168 & -0.28 & -0.26 & -0.30 & 57846 & 838 \\
Komodo 11.01 64-bit & 168 & 170 & 165 & -0.20 & -0.25 & -0.16 & 50260 & 799 \\
Rybka 3 32-bit & 165 & 164 & 167 & -0.15 & -0.05 & -0.26 & 195656 & 2823 \\
Torch v1 64-bit & 162 & 165 & 159 & -0.09 & -0.14 & -0.05 & 121710 & 1860 \\
NagaSkaki 4.00 & 162 & 164 & 160 & -0.23 & -0.27 & -0.19 & 17527 & 286 \\
Berserk 12 64-bit & 162 & 165 & 159 & -0.09 & -0.15 & -0.04 & 112855 & 1656 \\
Vitruvius 1.11C 32-bit & 162 & 164 & 160 & -0.14 & -0.16 & -0.12 & 97409 & 1453 \\
BBChess 1.3a & 160 & 161 & 160 & -0.13 & -0.12 & -0.14 & 34071 & 528 \\
Comet B68 & 159 & 161 & 157 & -0.10 & -0.14 & -0.06 & 66173 & 1043 \\ 
$\vdots$  & $\vdots$ & $\vdots$ & $\vdots$ & $\vdots$ & $\vdots$ & $\vdots$& $\vdots$ & $\vdots$ \\
Clarity 2.0.0 64-bit	& 140	& 141	& 139	& 0.28 & 	0.26& 	0.30& 	63446	& 1149 \\
\hline
\end{tabular}
\caption{Summary of computational results for the best versions of engines. The subscripts $W$ and $B$ refer to White and Black, respectively.}
\label{tab:all_engines_best_version}
\end{table}

\begin{table}[h!]
\centering
\begin{tabular}{|l|c|c|c|c|c|c|c|}
\hline
& Carlsen & Kasparov & Fischer & Kramnik & Anand & Karpov \\ \hline
Kasparov & 0.09 &  & & & &  \\ \hline
Fischer & 0.00 & 0.00 &  & & &  \\ \hline
Kramnik & 0.01 & 0.18 & 0.97 &  & &  \\ \hline
Anand & 0.00 & 0.07 & 0.93 & 0.30 &  &  \\ \hline
Karpov & 0.00 & 0.00 & 0.02 & 0.00 & 0.00 &   \\ \hline
Topalov & 0.00 & 0.00 & 0.00 & 0.00 & 0.00 & 0.11 \\ \hline
\end{tabular}
\caption{P-values under the null hypothesis that there is no difference in the MP distributions between the row player and the column player.}
\label{tab:mann-whitney_MP}
\end{table}
\end{document}